\documentclass{llncs}

\usepackage{graphicx}
\usepackage{bussproofs}
\usepackage{amsmath}
\usepackage{amssymb}

\newenvironment{bprooftree}
{\leavevmode\hbox\bgroup}
{\DisplayProof\egroup}

\title{On the Use of Computational Paths in Path Spaces of Homotopy Type Theory}

\author{Arthur F. Ramos\inst{1} \and Ruy J. G. B. de Queiroz\inst{1} \and Anjolina G. de Oliveira\inst{1} \and Tiago Mendon\c{c}a Lucena de Veras\inst{2}}

\institute{Centro de Inform\'atica\\
	Universidade Federal de Pernambuco\\
	\email{afr@cin.ufpe.br}\\
	\email{ruy@cin.ufpe.br}\\
	\email{ago@cin.ufpe.br}\\
	\and
	Departamento de Matem\'atica\\
	Universidade Federal Rural de Pernambuco\\
	\email{tiago.veras@ufrpe.br}
}

\begin{document}

	\maketitle
	
	\begin{abstract}
		The treatment of equality as a type in type theory gives rise to an interesting type-theoretic structure known as `identity type'. The idea is that, given terms $a,b$ of a type $A$, one may form the type $Id_{A}(a,b)$, whose elements are proofs that $a$ and $b$ are equal elements of type $A$. A term of this type, $p : Id_{A}(a,b)$, makes up for the grounds (or proof) that establishes that $a$ is indeed equal to $b$. Based on that, a proof of equality can be seen as a sequence of substitutions and rewrites, also known as a `computational path'. One interesting fact is that it is possible to rewrite computational paths using a set of reduction rules arising from an analysis of redundancies in paths. These rules were mapped by De Oliveira in 1994 in a term rewrite system known as $LND_{EQ}-TRS$. Here we use computational paths and this term rewrite system to work with path spaces. In homotopy type theory, the main technique used to define path spaces is the code-encode-decode approach. Our objective is to propose an alternative approach based on the theory of computational paths. We believe this new approach is simpler and more straightforward than the code-encode-decode one. We then use our approach to obtain two important results of homotopy type theory: the construction of the path space of the naturals and the calculation of the fundamental group of the circle.
	\end{abstract}
	
	\section{Introduction}\label{intro}
	
	There seems to be little doubt that the identity type is one of the most intriguing concepts of  Martin-L\"of's Type Theory. This claim is supported by recent groundbreaking discoveries. In 2005, Vladimir Voevodsky \cite{Vlad1} discovered the Univalent Models, resulting in a new area of research known as homotopy type theory \cite{Steve1}. This theory is based on the fact that a term of some identity type, for example $p: Id_{A}(a,b)$, has a clear homotopical interpretation. The interpretation is that the witness $p$ can be seen as a homotopical path between the points $a$ and $b$ within a topological space $A$. This simple interpretation has made clear the connection between type theory and homotopy theory, generating groundbreaking results, as one can see in \cite{hott,Steve1}. Nevertheless, it is important to emphasize that 
    the homotopic paths exist only in the semantic sense. In other words, there is no formal entity in type theory that represents these paths. They are not present in the syntax of type theory.
	
	In this work, we are interested in an entity known as computational path, originally proposed by \cite{Ruy4}. A computational path is an entity that establishes the equality between two terms of the same type. It differs from the homotopical path, since it is not only a semantic interpretation. It is a formal entity of the equality theory. In fact, we proposed in \cite{Ruy1} that it should be considered as the type of the identity type. Moreover, we have further developed this idea in \cite{Art3}, where we proposed a groupoid model and proved that computational paths also refute the uniqueness of identity proofs. Thus, we obtained a result  that is on par with the same one obtained by Hofmann \& Streicher (1995) for the original identity type \cite{hofmann2}.
	
	Our main idea in this work is to develop further our previous results. Specifically, we want to focus on the idea of path spaces. Since paths are not present in the syntax of traditional homotopy type theory, simulating the path of space of a complex entity could be a cumbersome task. The main approach is to use a technique known as code-encode-decode to try to simulate the path space \cite{hott}, making proofs somewhat difficult to read and understand. Since our theory adds the concept of path directly to type theory, working with path spaces is a simpler and more direct task. To support this claim, we will work with two examples. First, we will define the path space of the naturals, also showing that it is equivalent to the one obtained using the code-encode-decode technique. Secondly, we will provide a simple proof of that the fundamental group of the circle is isomorphic to the integers. Since the same proof using code-encode-decode is rather complex \cite{hott}, we hope to assert the advantage in simplicity of our approach. 
	
	\section{Computational Paths} \label{path}
	
	In this section, our objective is to give a brief introduction to the theory of computational paths. One should refer to \cite{Ruy1,Art3} for a detailed development of this theory. 
	
	A computational path is based on the idea that it is possible to formally define when two computational objects $a,b : A$ are equal. These two objects are equal if one can reach $b$ from $a$ applying a sequence of axioms or rules. This sequence of operations forms a path. Since it is between two computational objects, it is said that this path is a computational one. Also, an application of an axiom or a rule transforms (or rewrite) an term into another. For that reason, a computational path is also known as a sequence of rewrites. Nevertheless, before we define formally a computational path, we can take a look at one famous equality theory, the $\lambda\beta\eta-equality$ \cite{lambda}:
	
	\begin{definition}
		The \emph{$\lambda\beta\eta$-equality} is composed by the following axioms:
		
		\begin{enumerate}
			\item[$(\alpha)$] $\lambda x.M = \lambda y.M[y/x]$ \quad if $y \notin FV(M)$;
			\item[$(\beta)$] $(\lambda x.M)N = M[N/x]$;
			\item[$(\rho)$] $M = M$;
			\item[$(\eta)$] $(\lambda x.Mx) = M$ \quad $(x \notin FV(M))$.
		\end{enumerate}
		
		And the following rules of inference:

		\bigskip
		\noindent
		\begin{bprooftree}
			\AxiomC{$M = M'$ }
			\LeftLabel{$(\mu)$ \quad}
			\UnaryInfC{$NM = NM'$}
		\end{bprooftree}
		\begin{bprooftree}
			\AxiomC{$M = N$}
			\AxiomC{$N = P$}
			\LeftLabel{$(\tau)$}
			\BinaryInfC{$M = P$}
		\end{bprooftree}
		
		\bigskip
		\noindent
		\begin{bprooftree}
			\AxiomC{$M = M'$ }
			\LeftLabel{$(\nu)$ \quad}
			\UnaryInfC{$MN = M'N$}
		\end{bprooftree}
		\begin{bprooftree}
			\AxiomC{$M = N$}
			\LeftLabel{$(\sigma)$}
			\UnaryInfC{$N = M$}
		\end{bprooftree}
		
		\bigskip
		\noindent
		\begin{bprooftree}
			\AxiomC{$M = M'$ }
			\LeftLabel{$(\xi)$ \quad}
			\UnaryInfC{$\lambda x.M= \lambda x.M'$}
		\end{bprooftree}
		
		
		
		
	\end{definition}
	
	
	
	
	
	\begin{definition}($\beta$-equality \cite{lambda})
		$P$ is $\beta$-equal or $\beta$-convertible to $Q$  (notation $P=_\beta Q$)
		iff $Q$ is obtained from $P$ by a finite (perhaps empty)  series of $\beta$-contractions
		and reversed $\beta$-contractions  and changes of bound variables.  That is,
		$P=_\beta Q$ iff \textbf{there exist} $P_0, \ldots, P_n$ ($n\geq 0$)  such that
		$P_0\equiv P$,  $P_n\equiv Q$,
		$(\forall i\leq n-1) (P_i\triangleright_{1\beta}P_{i+1}  \mbox{ or }P_{i+1}\triangleright_{1\beta}P_i  \mbox{ or } P_i\equiv_\alpha P_{i+1}).$
	\end{definition}

	The same happens with $\lambda\beta\eta$-equality:\\
	\begin{definition}($\lambda\beta\eta$-equality \cite{lambda})
		The equality-relation determined by the theory $\lambda\beta\eta$ is called
		$=_{\beta\eta}$; that is, we define
		$$M=_{\beta\eta}N\quad\Leftrightarrow\quad\lambda\beta\eta\vdash M=N.$$
	\end{definition}
	
	\begin{example}
		Take the term $M\equiv(\lambda x.(\lambda y.yx)(\lambda w.zw))v$. Then, it is $\beta\eta$-equal to $N\equiv zv$ because of the sequence:\\
		$(\lambda x.(\lambda y.yx)(\lambda w.zw))v, \quad  (\lambda x.(\lambda y.yx)z)v, \quad   (\lambda y.yv)z , \quad zv$\\
		which starts from $M$ and ends with $N$, and each member of the sequence is obtained via 1-step $\beta$-contraction or $\eta$-contraction of a previous term in the sequence. To take this sequence into a {\em path\/}, one has to apply transitivity twice, as we we see below. Taking this sequence into a path leads us to the following:\\
	
		\noindent The first is equal to the second based on the grounds:\\
		$\eta((\lambda x.(\lambda y.yx)(\lambda w.zw))v,(\lambda x.(\lambda y.yx)z)v)$\\
		The second is equal to the third based on the grounds:\\
		$\beta((\lambda x.(\lambda y.yx)z)v,(\lambda y.yv)z)$\\
		Now, the first is equal to the third based on the grounds:\\
		$\tau(\eta((\lambda x.(\lambda y.yx)(\lambda w.zw))v,(\lambda x.(\lambda y.yx)z)v),\beta((\lambda x.(\lambda y.yx)z)v,(\lambda y.yv)z))$\\
		Now, the third is equal to the fourth one based on the grounds:\\
		$\beta((\lambda y.yv)z,zv)$\\
		Thus, the first one is equal to the fourth one based on the grounds:\\
		$\tau(\tau(\eta((\lambda x.(\lambda y.yx)(\lambda w.zw))v,(\lambda x.(\lambda y.yx)z)v),\beta((\lambda x.(\lambda y.yx)z)v,(\lambda y.yv)z)),\beta((\lambda y.yv)z,zv)))$.
	\end{example}
	
	
	The aforementioned theory establishes the equality between two $\lambda$-terms. Since we are working with computational objects as terms of a type, we need to translate the $\lambda\beta\eta$-equality to a suitable equality theory based on Martin L\"of's type theory. For the $\Pi$-type, for example, we obtain:
	
	\begin{definition}
		The equality theory of Martin L\"of's type theory has the following basic proof rules for the $\Pi$-type \cite{Ruy1,Art3}:
		
		\bigskip
		
		\noindent
		\begin{bprooftree}
			\hskip -0.3pt
			\alwaysNoLine
			\AxiomC{$N : A$}
			\AxiomC{$[x : A]$}
			\UnaryInfC{$M : B$}
			\alwaysSingleLine
			\LeftLabel{$(\beta$) \quad}
			\BinaryInfC{$(\lambda x.M)N = M[N/x] : B[N/x]$}
		\end{bprooftree}
		\begin{bprooftree}
			\hskip 17pt
			\alwaysNoLine
			\AxiomC{$[x : A]$}
			\UnaryInfC{$M = M' : B$}
			\alwaysSingleLine
			\LeftLabel{$(\xi)$ \quad}
			\UnaryInfC{$\lambda x.M = \lambda x.M' : \Pi_{(x : A)}B$}
		\end{bprooftree}
		
		\bigskip
		
		\noindent
		\begin{bprooftree}
			\hskip -0.5pt
			\AxiomC{$M : A$}
			\LeftLabel{$(\rho)$ \quad}
			\UnaryInfC{$M = M : A$}
		\end{bprooftree}
		\begin{bprooftree}
			\hskip 100pt
			\AxiomC{$M = M' : A$}
			\AxiomC{$N : \Pi_{(x : A)}B$}
			\LeftLabel{$(\mu)$ \quad}
			\BinaryInfC{$NM = NM' : B[M/x]$}
		\end{bprooftree}
		
		\bigskip
		
		\noindent
		\begin{bprooftree}
			\hskip -0.5pt
			\AxiomC{$M = N : A$}
			\LeftLabel{$(\sigma) \quad$}
			\UnaryInfC{$N = M : A$}
		\end{bprooftree}
		\begin{bprooftree}
			\hskip 105pt
			\AxiomC{$N : A$}
			\AxiomC{$M = M' : \Pi_{(x : A)}B$}
			\LeftLabel{$(\nu)$ \quad}
			\BinaryInfC{$MN = M'N : B[N/x]$}
		\end{bprooftree}
		
		\bigskip
		
		\noindent
		\begin{bprooftree}
			\hskip -0.5pt
			\AxiomC{$M = N : A$}
			\AxiomC{$N = P : A$}
			\LeftLabel{$(\tau)$ \quad}
			\BinaryInfC{$M = P : A$}
		\end{bprooftree}
		
		\bigskip
		
		\noindent
		\begin{bprooftree}
			\hskip -0.5pt
			\AxiomC{$M: \Pi_{(x : A)}B$}
			\LeftLabel{$(\eta)$ \quad}
			\RightLabel {$(x \notin FV(M))$}
			\UnaryInfC{$(\lambda x.Mx) = M: \Pi_{(x : A)}B$}
		\end{bprooftree}
		
		\bigskip
		
	\end{definition}
	
	We are finally able to formally define computational paths:
	
	\begin{definition}
		Let $a$ and $b$ be elements of a type $A$. Then, a \emph{computational path} $s$ from $a$ to $b$ is a composition of rewrites (each rewrite is an application of the inference rules of the equality theory of type theory or is a change of bound variables). We denote that by $a =_{s} b$.
	\end{definition}
	
	As we have seen in \textbf{example 1}, composition of rewrites are applications of the rule $\tau$. Since change of bound variables is possible, each term is considered up to $\alpha$-equivalence.
	
	\section{A Term Rewriting System for Paths}
	
	As we have just shown, a computational path establishes when two terms of the same type are equal. From the theory of computational paths, an interesting case arises. Suppose we have a path $s$ that establishes that $a =_{s} b : A$ and a path $t$ that establishes that $a =_{t} b : A$. Consider that $s$ and $t$ are formed by distinct compositions of rewrites. Is it possible to conclude that there are cases that $s$ and $t$ should be considered equivalent? The answer is \emph{yes}. Consider the following example:
	
	\begin{example}
		\noindent \normalfont Consider the path  $a =_{t} b : A$. By the symmetric property, we obtain $b =_{\sigma(t)} a : A$. What if we apply the property again on the path $\sigma(t)$? We would obtain a path  $a =_{\sigma(\sigma(t))} b : A$. Since we applied symmetry twice in succession, we obtained a path that is equivalent to the initial path $t$. For that reason, we conclude the act of applying symmetry twice in succession is a redundancy. We say that the path $\sigma(\sigma(t))$ can be reduced to the path $t$.
	\end{example}
	
	As one could see in the aforementioned example, different paths should be considered equal if one is just a redundant form of the other. The example that we have just seen is just a straightforward and simple case. Since the equality theory has a total of 7 axioms, the possibility of combinations that could generate redundancies are high. Fortunately, most possible redundancies were thoroughly mapped by \cite{Anjo1}. In that work, a system that establishes redundancies and creates rules that solve them was proposed. This system, known as $LND_{EQ}-TRS$, originally mapped a total of 39 rules. For each rule, there is a proof tree that constructs it. In this work, we have discovered $3$ entirely new rules and added to the system, making a total of $42$ rules. We included all rules in \textbf{appendix B}, highlighting the recently discovered ones. To illustrate those rules, take the case of \textbf{example 2}. We have the following \cite{Ruy1}:
	
	\bigskip
	\begin{prooftree}
		\AxiomC{$x =_{t} y : A$}
		\UnaryInfC{$y =_{\sigma(t)} x : A$}
		\RightLabel{\quad $\rhd_{ss}$ \quad $x =_{t} y : A$}
		\UnaryInfC{$x =_{\sigma(\sigma(t))} y : A$}
	\end{prooftree}
	
	\bigskip
	
	It is important to notice that we assign a label to every rule. In the previous case, we assigned the label $ss$.

	\begin{definition}($rw$-rule \cite{Art3})
		\normalfont An $rw$-rule is any of the rules defined in $LND_{EQ}-TRS$.
	\end{definition}
	
	\begin{definition}($rw$-contraction \cite{Art3})
		Let $s$ and $t$ be computational paths. We say that $s \rhd_{1rw} t$ (read as: $s$ $rw$-contracts to $t$) iff we can obtain $t$ from $s$ by an application of only one $rw$-rule. If $s$ can be reduced to $t$ by finite number of $rw$-contractions, then we say that $s \rhd_{rw} t$ (read as $s$ $rw$-reduces to $t$).
		
	\end{definition}
	
	\begin{definition}($rw$-equality \cite{Art3})
		\normalfont  Let $s$ and $t$ be computational paths. We say that $s =_{rw} t$ (read as: $s$ is $rw$-equal to $t$) iff $t$ can be obtained from $s$ by a finite (perhaps empty) series of $rw$-contractions and reversed $rw$-contractions. In other words, $s =_{rw} t$ iff there exists a sequence $R_{0},....,R_{n}$, with $n \geq 0$, such that
		
		\centering $(\forall i \leq n - 1) (R_{i}\rhd_{1rw} R_{i+1}$ or $R_{i+1} \rhd_{1rw} R_{i})$
		
		\centering  $R_{0} \equiv s$, \quad $R_{n} \equiv t$
	\end{definition}
	
	\begin{proposition}\label{proposition3.7} $rw$-equality  is transitive, symmetric and reflexive.
	\end{proposition}
	
	\begin{proof}
		Comes directly from the fact that $rw$-equality is the transitive, reflexive and symmetric closure of $rw$.
	\end{proof}
	
	The above proposition is rather important, since sometimes we want to work with paths up to $rw$-equality. For example, we can take a path $s$ and use it as a representative of an equivalence class, denoting this by $[s]_{rw}$.
	
	We'd like to mention that  $LND_{EQ}-TRS$ is terminating and confluent. The proof of this affirmation can be found in \cite{Anjo1,Ruy2,Ruy3,RuyAnjolinaLivro}.
	
	One should refer to \cite{Ruy1,RuyAnjolinaLivro} for a more complete and detailed explanation of the rules of $LND_{EQ}-TRS$.
	
	\section{Natural Numbers}
	Before proceeding to the actual content of this section, one should be familiar to the concept of transport. To see how it is defined in homotopy type theory, one could refer to \cite{hott}. Our formulation of transport is based on computational paths and thus, it is different from the traditional one. Since it is used in the proof of \textbf{theorem 3}, we included the details in \textbf{appendix C}.
	
	Our objective in this section is to use computational paths to define the path-space of the naturals. The Natural Numbers is a type defined inductively by an element $0 : \mathbb N$ and a function $succ : \mathbb N \rightarrow \mathbb N$. Instead of having to use the rigid syntax of type theory to simulate the path space, in our approach the path space of the naturals is characterized inductively. The basis is the reflexive path $0 =_{\rho} 0$.  All subsequent paths are constructed by applications of the inference rules of $\lambda\beta\eta$-equality. In fact, we can now show that this characterization is similar to the one constructed in \cite{hott}. To do this, we use code-encode-decode. For $\mathbb N$, we define $code$ recursively \cite{hott}:
	
	\begin{center}
		$code(0,0) \equiv 1$\\
		$code(succ(m),0) \equiv 0$\\
		$code(0,succ(m)) \equiv 0$\\
		$code(succ(m),succ(n)) \equiv code (m,n)$
	\end{center}
	
	We also define a dependent function $r : \Pi_{(n : \mathbb N)}code(m,n)$, with:
	
	\begin{center}
		$r(0) \equiv *$\\
		$r(succ(n)) \equiv r(n)$
	\end{center}
	
	Before we show results directly related to $\mathbb N$, we need the following result:
	
	\begin{theorem}
		For any type A and a path $x =_{\rho} x : A$, if a path $s$ is obtained by a series (perhaps empty) of applications of axioms and rules of inference of $\lambda\beta\eta$-equality theory for type theory to the path $\rho$, then there is a path $t'$ such that $s =_{t'} \rho$.
	\end{theorem}
	
	\begin{proof}
		In \textbf{appendix D}.
	\end{proof}
	
	\begin{theorem}
		For any $m,n  : \mathbb N$, if there is a path $m =_{t} n : \mathbb N$, then $t \rhd \rho$.
	\end{theorem}
	
	\begin{proof}
		Since all paths are constructed from the reflexive path $0 =_{\rho} 0$, this is a direct application of \textbf{theorem 1}.
	\end{proof}
	
	\begin{theorem}
		For any $m,n : \mathbb{N}$, we have $(m = n) \simeq code(m,n)$
	\end{theorem}
	
	\begin{proof}
		We need to define $encode$ and $decode$ and prove that they are pseudo-inverses. We define $encode : \Pi_{(m,n : \mathbb N)}(m = n) \rightarrow code(m,n)$ as:
		
		\begin{center}
			$encode(m,n,p) \equiv transport^{code(m,-)} (p, r(m))$
		\end{center}
		
		We define $decode : \Pi_{(m,n : \mathbb N)}code(m,n) \rightarrow (m = n)$ recursively:
		
		\begin{center}
			$decode(0,0, c) \equiv 0 =_{\rho} 0$\\
			$decode(succ(m), 0, c) \equiv 0$\\
			$decode(0, succ(m), c) \equiv 0$\\
			$decode(succ(m), succ(n), c) \equiv \mu_{succ}(decode(m,n,c))$
		\end{center}
		
		We now prove that if $m =_{p} n$, then $decode(code(m,n)) = \rho$. We prove by induction. The base is trivial, since $decode(0,0,c) \equiv \rho$. Now, consider $decode(succ(m),succ(n),c)$. We have that $decode(succ(m),succ(n),c) \equiv \mu_{succ}(decode(m,n,c))$. By the inductive hypothesis, $decode(m,n,c) \equiv \rho$. Thus, we need to prove that $\mu_{succ} = \rho$. This last step is a straightforward application of \textbf{rule 40} (check \textbf{appendix B}).  Therefore,  $\mu_{succ} =_{mxp} \rho$. With this information, we can start the proof of the equivalence.
		
		For any $m =_{p} n$, we have:
		
		\begin{center}
			$encode(m,n,p) \equiv transport^{code(m,-)}(p,r(m))$
		\end{center}
		
		Thus (check \textbf{appendix C}):
		
		\bigskip
		
		\begin{prooftree}
			\AxiomC{$m =_{p} n$}
			\AxiomC{$r(m) : code(m,m)$}
			\RightLabel{$=_{\mu(p)} \quad (r(n) : code(m,n))$}
			\BinaryInfC{$p(m,n) \circ r(m) : code(m,n)$}
		\end{prooftree}
		
		\bigskip
		
		Now, we know that $decode(r(n) : code(m,n)) = \rho$ and, by \textbf{theorem 2}, $p = \rho$.
		
		The proof starting from a $c : code(m,n)$ is equal to the one presented in \cite{hott}. We prove by induction. If $m$ and $n$ are $0$, we have the trivial path $0 =_{\rho} 0$, thus $decode(0,0,c) = \rho_{0}$, whereas $encode(0,0,\rho_{0}) \equiv r(0) \equiv *$. Now, we recall that the induction for the unit type is given by $* \rhd_{\eta} x : 1$ \cite{hott}. Thus, we conclude that every $x : 1$ is equal to $*$, since we have $ x =_{\sigma(\eta)} * : 1$. In the case of $decode(succ(m),0,c)$ or $decode(0,succ(n),c)$, we have $c : 0$. The only case left is for $decode(succ(m), succ(n),c)$. Similar to \cite{hott}, we prove by induction:
		
		\bigskip
		
		$encode(succ(m),succ(n), decode(succ(m),succ(n),c))$
		
		\quad \quad \quad \quad \quad \quad \quad \quad \quad \quad $= encode(succ(m),succ(n), \mu_{succ}(decode(m,n,c))$
		
		\quad \quad \quad \quad \quad \quad \quad \quad \quad \quad $= transport^{code(succ(m),-)}(\mu_{succ}(decode(m,n,c)),r(succ(m))$
		
		\quad \quad \quad \quad \quad \quad \quad \quad \quad \quad $= transport^{code(succ(m),succ(-)}(decode(m,n,c),r(succ(m)))$
		
		\quad \quad \quad \quad \quad \quad \quad \quad \quad \quad $= transport^{code(m,-)}(decode(m,n,c),r(m))$
		
		\quad \quad \quad \quad \quad \quad \quad \quad \quad \quad $= encode(m,n,decode(m,n,c))$
		
		\quad \quad \quad \quad \quad \quad \quad \quad \quad \quad $= c$
	\end{proof}
	
	Therefore, we conclude that our simple inductive definition is equivalent to all the machinery of the code-encode-decode one.
	
	\section{Fundamental Group of a Circle}
	
	The objective of this section is to show that it is possible to use computational paths to obtain one of the main results of homotopy theory, the fact that the fundamental group of a circle is isomorphic to the integers group. We avoid again the use of the heavy and rather complicated machinery of the code-encode-decode approach. First, we define the circle as follows:
	
	\begin{definition}[The circle $S^1$]
		The circle is the type generated by:
		
		\begin{itemize}
			\item A point - $base : S^1$	
			\item A computational path - $base =_{loop} base : S^1$.
		\end{itemize}
	\end{definition}
	
	The first thing one should notice is that this definition doest not use only the points of the type $S^1$, but also a computational path $loop$ between those points. That is way it is called a higher inductive type \cite{hott}. Our approach differs from the classic one on the fact that we do not need to simulate the path-space between those points, since computational paths exist in the syntax of the theory. Thus, if one starts with a path  $base =_{loop} base : S^1$, one can naturally obtain additional paths applying the path-axioms $\rho$, $\tau$ and $\sigma$.  Thus, one has a path $\sigma(loop) = loop^{-1}$, $\tau(loop, loop)$, etc. In classic type theory, the existence of those additional paths comes from establishing that the paths should be freely generated by the constructors \cite{hott}. In our approach, we do not have to appeal for this kind of argument, since all paths comes naturally from direct applications of the axioms.
	
	With that in mind, one can define the fundamental group of a circle. In homotopy theory, the fundamental group is the one formed by all equivalence classes up to homotopy of paths (loops) starting from a point $a$ and also ending at $a$. Since the we use computational paths as the syntax counterpart of homotopic paths in type theory, we use it to propose the following definition:
	
	\begin{definition}[$\Pi_{1}(A,a)$ structure]
		$\Pi_{1}(A,a)$ is a structure defined as follows:
		
		\begin{center}
			$\Pi_{1}(A, a) = \{[loop]_{rw} \mid a =_{loop} a: A\}$
		\end{center}
	\end{definition}
	
	We use this structure to define the fundamental group of a circle. We also need to show that it is indeed a group.
	
	\begin{proposition}
		$(\Pi_{1}(S,a), \circ)$ is a group.
	\end{proposition}
	
	\begin{proof}
		The first thing to define is the group operation $\circ$. Given any $a =_{r} a : S^1$ and $a =_{t} a : S^1$, we define $r \circ s$ as $\tau(s,r)$. Thus, we now need to check the group conditions:
		
		\begin{itemize}
			
			\item Closure: Given $a =_{r} a : S^1$ and $a =_{t} a : S^1$, $r \circ s$ must be a member of the group. Indeed, $r \circ s = \tau(s,r)$ is a computational path $a =_{\tau(s,r)} a : S^1$.
			\bigskip
			\item Inverse: Every member of the group must have an inverse. Indeed, if we have a path $r$, we can apply $\sigma(r)$. We claim that $\sigma(r)$ is the inverse of $r$, since we have:
			
			\begin{center}
				$\sigma(r) \circ r = \tau(r, \sigma(r)) =_{tr} \rho$
				
				$r \circ \sigma(r) = \tau(\sigma(r), r) =_{tsr} \rho$
			\end{center}
			
			Since we are working up to $rw$-equality, the equalities hold strictly.
			
			\item Identity: We use the path $a =_{\rho} a : S^1$ as the identity. Indeed, we have:
			
			\begin{center}
				$r \circ \rho = \tau(\rho,r) =_{tlr} r$
				
				$\rho \circ r = \tau(r,\rho) =_{trr} r$.
			\end{center}
			
			\item Associativity: Given any members of the group $a =_{r} a : S^1$, $a =_{t} a$ and $a =_{s} a$, we want that $r \circ (s \circ t) = (r \circ s) \circ t$:
			
			\begin{center}
				$r \circ (s \circ t) = \tau(\tau(t,s), r) =_{tt} \tau(t,\tau(s,r)) = (r \circ s) \circ t$
			\end{center}
			
		\end{itemize}
		
		All conditions have been satisfied. $(\Pi_{1}(S,a), \circ)$ is a group.
	\end{proof}
	
	Thus, 	$(\Pi_{1}(S,a), \circ)$ is indeed a group. We call this group the fundamental group of $S^1$. Therefore, the objective of this section is to show that $\Pi_{1}(S,a) \simeq \mathbb{Z}$.
	
	Before we start developing this proof, the following lemma will prove to be useful:
	
	\begin{lemma}
		All paths generated by a path $a =_{loop} a$ are $rw$-equal to a path $loop^{n}$, for a $n \in \mathbb Z$.
	\end{lemma}
	
	We have said that from a $loop$, one freely generate different paths applying the composition $\tau$ and the symmetry. Thus, one can, for example, obtain something such as $loop \circ loop \circ loop^{-1} \circ loop...$. Our objective with this lemma is to show that, in fact, this path can be reduced to a path of the form $loop^{n}$, for $n \in \mathbb Z$.
	
	\begin{proof}
		The idea is to proceed by induction. We start from a base $\rho$. For the base case, it is trivially true, since we define it to be equal to $loop^{0}$. From $\rho$, one can construct more complex paths by composing with $loop$ or $\sigma(loop)$ on each step.
		We have the following induction steps:
		
		\begin{itemize}
			\item A path of the form $\rho$ concatenated with $loop$: We have $\rho \circ loop = \tau(loop,\rho) =_{trr} loop = loop^{1}$;
			\bigskip
			\item A path of the form $\rho$ concatenated with $\sigma(loop)$: We have $\rho \circ \sigma(loop) = \tau(\sigma(loop),\rho) =_{trr} = \sigma(loop) = loop^{-1}$
			\bigskip
			\item A path of the form $loop^{n}$ concatenated with $loop$: We have $loop^{n} \circ loop = loop^{n+1}$.
			\bigskip
			\item A path of the form $loop^{n}$ concatenated with $\sigma(loop)$: We have $loop^{n} \circ \sigma(loop)$ $= (loop^{n-1} \circ loop) \circ \sigma(loop) =_{tt} loop^{n-1} \circ (loop \circ \sigma(loop)) =$ $loop^{n-1} \circ (\tau(\sigma(loop), loop)) =_{tsr} = loop^{n-1} \circ \rho = \tau(\rho, loop^{n-1}) =_{tlr} loop^{n-1}$
			\bigskip
			\item A path of the form $loop^{-n}$ concatenated with $loop$: We have $loop^{-n}$ = $loop^{-(n - 1)} \circ loop^{-1} = loop^{-(n - 1)} \circ \sigma(loop)$. Thus, we have $(loop^{-(n - 1)} \circ \sigma(loop)) \circ loop$ $=_{tt}$ $loop^{-(n - 1)} \circ (\sigma(loop) \circ loop)$ $=$ $loop^{-(n-1)} \circ \tau(loop,\sigma(loop)) =_{tr}$ $=$ $loop^{-(n-1)} \circ \rho = \tau(\rho,loop^{-(n-1)}) =_{tlr} loop^{-(n-1)}$.
			\bigskip
			\item a path of the form $loop^{-n}$ concatenated with $\sigma(loop)$: We have $loop^{-n} \circ loop^{-1} = loop^{-(n + 1)}$ 
		\end{itemize}
		
		Thus, every path is of the form $loop^{n}$, with $n \in \mathbb Z$.
	\end{proof}

	This lemma shows that every path of the fundamental group can be represented by a path of the form $loop^{n}$, with $n \in \mathbb Z$.
	
	\begin{theorem}
		$\Pi_{1}(S,a) \simeq \mathbb{Z}$
	\end{theorem}
	
	To prove this theorem, one could use the approach proposed in \cite{hott}, defining an encode and decode functions. Nevertheless, since our computational paths are part of the syntax, one does not need to rely on this kind of approach to simulate a path-space, we can work directly with the concept of path.
	
	\begin{proof}
		The proof is done by establishing a function from $\Pi_{1}(S,a)$ to $\mathbb{Z}$ and then an inverse from $\mathbb{Z}$ to $\Pi_{1}(S,a)$. Since we have access to the previous lemma, this task is not too difficult. The main idea is that the $n$ on $loop^{n}$ means the amount of times one goes around the circle, while the sign gives the direction (clockwise or anti-clockwise). In other words, it is the $winding$ number. Since we have shown that every path of the fundamental group is of the form $loop^{n}$, with $n \in \mathbb Z$, then we just need to translate $loop^{n}$ to an integer $n$ and an integer $n$ to a path $loop^{n}$. We define two functions, $toInteger: \Pi_{1}(S,a) \rightarrow \mathbb Z$ and $toPath: \mathbb Z \rightarrow \Pi_{1}(S,a)$:
		
		\begin{itemize}
			\item $toInteger$: To define this function, we use the help of two functions defined in $\mathbb Z$: the successor function $succ$ and the predecessor function $pred$. We define $toInteger$ as follows. Of course, we use directly the fact that every path of $\Pi_{1}(S,a)$ is of the form $loop^{n}$ with $n \in \mathbb Z$:
			
			\begin{equation*}
			toInteger: \begin{cases}
			toInteger(loop^n \equiv \rho) = 0 \quad \quad \quad \quad \quad \quad \quad \quad \quad \quad \enskip n = 0          \\
			toInteger(loop^{n}) = succ(toInteger(loop^{n-1})) \quad \quad n > 0  \\
			toInteger(loop^{n}) = pred(toInteger(loop^{n+1})) \quad \quad n < 0 \\
			\end{cases}
			\end{equation*}
			
			\item $toPath$: We just need to transform an integer $n$ into a path $loop^{n}$:
			
			\begin{equation*}
			toPath: \begin{cases}
			toPath(n) = \rho \quad \quad \quad \quad \quad \quad \quad \quad \quad\quad \enskip n = 0 \\
			toPath(n) = toPath(n - 1) \circ loop \quad \quad n > 0 \\
			toPath(n) = toPath(n + 1) \circ \sigma(loop) \quad n < 0 \\
			\end{cases}
			\end{equation*}
		\end{itemize}
		
		That they are inverses is a straightforward check. Therefore, we have $\Pi_{1}(S,a) \simeq \mathbb{Z}$.
	\end{proof}
	
	\section{Conclusion}
	
	Based on the theory of computational paths and the term rewriting system, we have proposed the main topic of this work: it is possible to use this approach to define path spaces directly, avoiding the use of complex techniques such as code-encode-decode. We have seen that the reason for that is the fact that the concept of path is only present semantically in traditional homotopy type theory, whereas in our approach it is added directly to the syntax. To illustrate our point, we have focused on two examples: the path space of the naturals and the proof that the fundamental group of the circle is isomorphic to the integers. In the first example, we have shown a straightforward inductive definition of the path space, showing that it is equivalent to the code-encode-decode one. In the second example, we have shown an easy proof for the fundamental group of the circle, using only basic rewriting rules in the process.       
	
	\bibliographystyle{plain}
	\bibliography{ref1}	
	
	\newpage
	\appendix
	
	\section{Subterm Substitution}
	
	In Equational Logic, the sub-term substitution is given by the following inference rule \cite{Ruy2}:
	\begin{center}
		\begin{bprooftree}
			\AxiomC{$s = t$ }
			\UnaryInfC{$s\theta = t\theta$}
		\end{bprooftree}
	\end{center}
	
	One problem is that such rule does not respect the sub-formula property. To deal with that, \cite{chenadec} proposes two inference rules:
	
	\begin{center}
		\begin{bprooftree}
			\AxiomC{$M = N$}
			\AxiomC{$C[N] = O$}
			\RightLabel{$IL$ \quad}
			\BinaryInfC{$C[M] = O$}
		\end{bprooftree}
		\begin{bprooftree}
			\AxiomC{$M = C[N]$}
			\AxiomC{$N = O$}
			\RightLabel{$IR$ \quad}
			\BinaryInfC{$M = C[O]$}
		\end{bprooftree}
	\end{center}
	
	\noindent where M, N and O are terms.
	
	As proposed in \cite{Ruy1}, we can define similar rules using computational paths, as follows:
	
	\begin{center}
		\begin{bprooftree}
			\AxiomC{$x =_r {\cal C}[y]: A$}
			\AxiomC{$y =_s u : A'$}
			\BinaryInfC{$x =_{{\tt sub}_{\tt L}(r,s)} {\cal C}[u]: A$}
		\end{bprooftree}
		\begin{bprooftree}
			\AxiomC{$x =_r w : A'$}
			\AxiomC{${\cal C}[w]=_s u : A$}
			\BinaryInfC{${\cal C}[x]=_{{\tt sub}_{\tt R}(r,s)} u : A$}
		\end{bprooftree}
	\end{center}
	
	\noindent where $C$ is the context in which the sub-term detached by '[ ]' appears and $A'$ could be a sub-domain of $A$, equal to $A$ or disjoint to $A$.
	
	In the rule above, ${\cal C}[u]$ should be understood as the result of replacing every occurrence of $y$ by $u$ in $C$.
	
	\section{List of Rewrite Rules}
	
	We present all rewrite rules of $LND_{EQ}-TRS$. They are as follows (All but the last three have been taken from \cite{Ruy1}):
	\\
	
	\noindent 1. $\sigma(\rho) \triangleright_{sr} \rho$ \\ 
	2. $\sigma(\sigma(r)) \triangleright_{ss} r$\\ 
	3. $\tau({\cal C}[r] , {\cal C}[\sigma(r)]) \triangleright_{tr}  {\cal C }[\rho]$\\ 
	4. $\tau({\cal C}[\sigma(r)], {\cal C}[r]) \triangleright_{tsr} {\cal C}[\rho]$\\ 
	5. $\tau({\cal C}[r], {\cal C}[\rho]) \triangleright_{trr} {\cal C}[r]$\\ 
	6. $\tau({\cal C}[\rho], {\cal C}[r]) \triangleright_{tlr} {\cal C}[r]$ \\ 
	7. ${\tt sub_L}({\cal C}[r], {\cal C}[\rho]) \triangleright_{slr} {\cal C}[r]$\\ 
	8. ${\tt sub_R}({\cal C}[\rho], {\cal C}[r]) \triangleright_{srr} {\cal C}[r]$ \\
	9. ${\tt sub_L} ({\tt sub_L} (s, {\cal C}[r]), {\cal C}[\sigma(r)]) \triangleright_{sls} s$\\
	10. ${\tt sub_L} ( {\tt sub_L} (s , {\cal C}[\sigma(r)]) , {\cal C}[r]) \triangleright_{slss} s$\\ 
	11. ${\tt sub_R} ({\cal C}[s], {\tt sub_R} ({\cal C}[\sigma(s)],r)) \triangleright_{srs} r$\\ 
	12. ${\tt sub_R} ({\cal C}[\sigma(s)], {\tt sub_R} ({\cal C}[s] ,  r )) \triangleright_{srrr} r$\\ 
	13. 
	$\mu_1 ( \xi_1 ( r))\triangleright_{mx2l1} r$\\
	14. $\mu_1 ( \xi_\land ( r,s))\triangleright_{mx2l2} r$\\
	15.
	$\mu_2 ( \xi_\land ( r,s))\triangleright_{mx2r1} s$\\
	16.
	$\mu_2 ( \xi_2 ( s))\triangleright_{mx2r2} s$\\
	17. 
	$\mu ( \xi_1 (r) , s , u) \triangleright_{mx3l} s$\\ 
	18. 
	$\mu (\xi_2 (r) , s , u) \triangleright_{mx3r} u$\\ 
	19.
	$\nu (\xi (r)) \triangleright_{mxl} r$\\ 
	20.
	$\mu (\xi_2 (r) , s) \triangleright_{mxr} s$\\ 
	21.
	$\xi ( \mu_1 (r),\mu_2(r) ) \triangleright_{mx} r$ \\ 
	22.
	$\mu ( t, \xi_1 (r), \xi_2 (s)) \triangleright_{mxx} t$ \\ 
	23. 
	$\xi ( \nu (r) ) \triangleright_{xmr} r$ \\ 
	24. 
	$\mu (s,\xi_2 (r)) \triangleright_{mx1r} s$\\ 
	25. $\sigma(\tau(r,s)) \triangleright_{stss} \tau(\sigma(s),  \sigma(r))$\\ 
	26. $\sigma({\tt sub_L}(r,s)) \triangleright_{ssbl} {\tt sub_R}(\sigma(s), \sigma(r))$\\ 
	27. $\sigma ({\tt sub_R} (r,s)) \triangleright_{ssbr} {\tt sub_L} (\sigma
	(s),  \sigma (r))$\\ 
	28. $\sigma(\xi (r)) \triangleright_{sx} \xi ( \sigma(r))$\\ 
	29. $\sigma(\xi (s, r)) \triangleright_{sxss} \xi ( \sigma(s),  \sigma(r))$\\ 
	30. $\sigma(\mu (r)) \triangleright_{sm} \mu ( \sigma(r))$\\ 
	31. $\sigma(\mu (s, r)) \triangleright_{smss} \mu (\sigma(s),  \sigma(r))$\\ 
	32. $\sigma(\mu (r,u,v)) \triangleright_{smsss} \mu ( \sigma(r),\sigma(u),\sigma(v))$\\
	33. $\tau (r, {\tt sub_L} (\rho , s)) \triangleright_{tsbll} {\tt sub_L}  (r,s)$\\ 
	34. $\tau (r, {\tt sub_R} (s, \rho)) \triangleright_{tsbrl}  {\tt 
		sub_L} (r,s)$\\ 
	35. $\tau({\tt sub_L}(r,s),t) \triangleright_{tsblr} \tau (r, {\tt 
		sub_R} (s,t))$\\ 
	36. $\tau ({\tt sub_R} (s,t),u) \triangleright_{tsbrr} {\tt sub_R} (s, \tau  (t,u))$\\ 
	37. $\tau(\tau(t,r),s) \triangleright_{tt} \tau(t,\tau (r,s)) $\\
	38. $\tau ({\cal C}[u], \tau ({\cal C}[\sigma(u)] , v)) \triangleright_{tts} v$\\
	39. $\tau ({\cal C}[\sigma(u)] , \tau ({\cal C}[u] , v)) \triangleright_{tst} u$\\
	40. $\mu_{f}(\rho_{x}) =_{mxp} \rho_{f(x)}$\\
	41. $\nu(\rho_{x}) =_{nxp} \rho_{f(x)}$\\
	42. $\xi(\rho) =_{xxp} \rho$.
	
	\bigskip
	
	Rules $40$, $41$ and $42$ have been recently discovered and appears for the first time in this work. They come form the following derivation trees:
	
	\bigskip
	
	Rule $40$:
	
	\bigskip
	
	\begin{prooftree}
		\AxiomC{$x =_{\rho_{x}} x : A$}
		\AxiomC{$[f : A \rightarrow B]$}
		\RightLabel{$\rhd_{mxp}$ \quad $f(x) =_{\rho_{f(x)}} f(x) : B(x)$}
		\BinaryInfC{$f(x) =_{\mu(\rho_{x})} f(x) : B(x)$}
	\end{prooftree}
	
	\bigskip
	
	Rule $41$:
	
	\bigskip
	
	\begin{prooftree}
		\AxiomC{$f =_{\rho} f : \Pi_{(x:A)}B(x)$}
		\RightLabel{$\rhd_{nxp}$ \quad $f(x) =_{\rho_{f(x)}} f(x)$}
		\UnaryInfC{$f(x) =_{\nu(\rho_{x})} f(x) : B(x)$}		
	\end{prooftree} 
	
	\bigskip

	Rule $42$:
	
	\begin{prooftree}
		\AxiomC{$b(x) =_{\rho} b(x) : B$}
		\AxiomC{$x : A$}
		\RightLabel{$\rhd_{xxp}$ \quad $\lambda x.b(x) =_{\rho} \lambda x.b(x)$}
		\BinaryInfC{$\lambda x.b(x) =_{\xi(\rho)} \lambda x.b(x) : A \rightarrow B$}	
	\end{prooftree}
	
	\bigskip
	
	\newpage
	
	\section{Transport}
	
	As stated in  \cite{Ruy5}, substitution can take place when no quantifier is involved. In this sense, there is a 'quantifier-less' notion of substitution. In type theory, this 'quantifier-less' substitution is given by a operation known as transport \cite{hott}. In our path-based approach, we formulate a new inference rule of 'quantifier-less' substitution \cite{Ruy5}:
	
	\bigskip
	\begin{prooftree}
		
		\AxiomC{$x =_{p} y : A$}
		\AxiomC{$f(x) : P(x)$}
		\BinaryInfC{$p(x,y)\circ f(x) : P(y)$}
		
	\end{prooftree}
	
	\bigskip
	
	We use this transport operation to solve one essential issue of our path-based approach. We know that given a path $x =_{p} y : A$ and function $f: A \rightarrow B$, the application of axiom $\mu$ yields the path $f(x) =_{\mu_{f}(p)} f(y) : B$. The problem arises when we try to apply the same axiom for a dependent function $f : \Pi_{(x : A)} P(x)$. In that case, we want $f(x) = f(y)$, but we cannot guarantee that the type of $f(x) : P(x)$ is the same as $f(y) : P(y)$. The solution is to apply the transport operation and thus, we can guarantee that the types are the same:
	
	\begin{center}
		
		\bigskip
		\begin{prooftree}
			\AxiomC{$x =_{p} y : A$}
			\AxiomC{$f : \Pi_{(x : A)} P(x)$}
			\BinaryInfC{$p(x,y) \circ f(x) =_{\mu_{f}(p)} f(y) : P(y)$}
		\end{prooftree}
	\end{center}
	
	\bigskip
	
	\newpage
	
	\section{Reflexivity}
	
	Here we provide a proof for \textbf{Theorem 1}, i.e., for the following statement: for any type A and a path $x =_{\rho} x : A$, if a path $s$ is obtained by a series (perhaps empty) of applications of axioms and rules of inference of $\lambda\beta\eta$-equality theory for type theory to the path $\rho$, then there is a path $t'$ such that $s =_{t'} \rho$.
	
	\begin{proof}
		The base case is straightforward. We can start only with a path $x =_{\rho} x$. In that case, it is easy, since we have $\rho =_{\rho} \rho$.
		
		Now, we consider the inductive steps. Starting from a path $s$ and applying $\tau$, $\sigma$, we already have rules yield the desired path:
		
		\begin{itemize}
			\item $s = \sigma(s')$, with $s' =_{t'} \rho$.
			
			In this case, we have $s = \sigma(s') = \sigma(\rho) =_{sr} \rho$. 
			
			\bigskip
			
			\item $s = \tau(s',s'')$, with $s' =_{t'} \rho$ and $s'' =_{t''} \rho$.
			
			We have that $s = \tau(s',s'') = \tau(\rho,\rho) =_{trr} \rho$ 
			
		\end{itemize}
		
		The cases for applications of $\mu$, $\nu$ and $\xi$ remain to be proved: 
		
		\begin{itemize}

			\item $s = \mu(s')$, with $s' =_{t'} \rho$.	
			
			We use \textbf{rule 40}. We have $s = \mu(s') = \mu(\rho) =_{mxp} \rho$.
			
			\bigskip
			
			\item $s = \nu(s')$, with $s' =_{t'} \rho$.

			We use \textbf{rule 41}. We have $s = \nu(s') = \nu(\rho) =_{nxp} \rho$.
			
			\bigskip
			
			\item $s = \xi(s')$, with $s' =_{t'} \rho$.
			
			We use \textbf{rule 42}. We have $s = \xi(s') = \xi(\rho) =_{xxp} \rho$.
			
			\bigskip
			
		\end{itemize}
	\end{proof}

\end{document}